\documentclass[12pt,a4paper]{article}

\RequirePackage[l2tabu, orthodox]{nag}

\usepackage{mjheppub}
\usepackage{color}
\usepackage{lipsum}
\usepackage{amsthm,amssymb,amsmath,epic,eepic,float}
\usepackage{rotating,epsfig,indentfirst,array,varioref}
\usepackage{appendix,marginnote,tikz,pgf,mathtools}
\usepackage{setspace}

\usepackage{tikz}
\usetikzlibrary{decorations.pathreplacing}
\usetikzlibrary{calc}

\newtheorem{prop}{Proposition}
\newtheorem{lemma}{Lemma}

\usepackage{mjheppub}
\usepackage[mathscr]{eucal}
\usepackage{amsmath,amsfonts,amssymb,amsthm}
\usepackage[matrix,arrow]{xy}
\usepackage{amsthm,amsmath,latexsym,amssymb,amsfonts,amsbsy}
\usepackage{wrapfig}
\usepackage{graphicx}
\usepackage{float}
\usepackage{amsmath}
\usepackage{mathtools}
\usepackage{amssymb}
\usepackage{braket}
\usepackage[T1]{fontenc}

\makeatletter
\def\@@bfil{\leaders \vrule \@height \ht\z@ \@depth \z@ \hfill}
\def\@bLfil{\@@bfil}
\def\@bRfil{\@@bfil}
\def\resetbraceratio{\gdef\@bLfil{\@@bfil}\gdef\@bRfil{\@@bfil}}
\def\setbraceratio#1#2{
  \let\@bLfil\relax
  \multido{\iA=1+1}{#1}{\gappto\@bLfil{\@@bfil}}
  \let\@bRfil\relax
  \multido{\iA=1+1}{#2}{\gappto\@bRfil{\@@bfil}}
}
\def\upbracefill{$\m@th\setbox\z@\hbox{$\braceld$}\bracelu\@bLfil\bracerd\braceld\@bRfil\braceru$}
\def\downbracefill{$\m@th\setbox\z@\hbox{$\braceld$}\braceld\@bLfil\braceru\bracelu\@bRfil\bracerd$}
\makeatother

\usepackage{tikz,pgf}
\usetikzlibrary{shapes}
\usetikzlibrary{calc}
\usetikzlibrary{decorations.pathmorphing}
\usetikzlibrary{decorations.pathreplacing,shapes.misc}
\usetikzlibrary{positioning}
\usetikzlibrary{arrows}
\usetikzlibrary{decorations.markings}

\def\be{\begin{equation}}
\def\ee{\end{equation}}
\def\barr{\begin{array}}
\def\earr{\end{array}}

\def\1{\tilde{1}}
\def\2{\tilde{2}}
\def\3{\tilde{3}}

\newcommand{\ba}{\begin{equation}\begin{aligned}}
\newcommand{\ea}{\end{aligned}\end{equation}}

\newcommand{\bml}{\begin{multline}}
\newcommand{\eml}{\end{multline}}

\newcommand{\CC}{\mathbb{C}}
\newcommand{\RR}{\mathbb{R}}

\newcommand{\ZZ}{\mathbb{Z}}
\newcommand{\PP}{\mathbb{P}}

\newcommand{\dd}{\mathrm{d}}
\newcommand{\pd}{\partial}

\newcommand{\Res}{\mathrm{Res}}

\newcommand{\MM}{\mathscr{M}}
\newcommand{\aA}{\mathscr{A}}

\mathtoolsset{showonlyrefs}

\begin{document}

\title{\boldmath 
Special  geometry on  Calabi--Yau moduli spaces 
and    $Q$--invariant Milnor rings \let\thefootnote\relax\footnote{\textdagger Contribution to
	Proceedings of  International Congress of Mathematicians 2018,
	Rio de Janeiro,(2018)}} 
\author{Alexander Belavin}

\affiliation{ L.D. Landau Institute for Theoretical Physics\\
  Semenov av. 1-A\\ Chernogolovka, 142432  Moscow region, Russia}

\affiliation{$^2$ Department of Quantum Physics\\ 
Institute for Information Transmission Problems\\
 Bolshoy Karetny per. 19, 127994 Moscow, Russia}




\emailAdd{sashabelavin@gmail.com}

\abstract{The moduli spaces of  Calabi--Yau (CY) manifolds are the special
  K\"ahler manifolds.
The special  K\"ahler geometry determines the low-energy effective theory which arises in 
Superstring theory after the compactification on a CY manifold. For the cases, where the CY manifold is given as a hypersurface in the weighted projective space,  a new procedure  for  computing the K\"ahler potential
 of the moduli space has been proposed in \cite {AKBA1,AKBA2, AKBA3}. The method is based on  the fact  that the moduli space of CY manifolds is a marginal subspace of the Frobenius manifold which arises on the deformation space of the corresponding Landau--Ginzburg superpotential. I review  this approach and demonstrate its efficiency by computing 
the Special geometry  of the 101-dimensional moduli space of the quintic threefold around the orbifold point  \cite {AKBA3}. }

\maketitle
\flushbottom

\section{ Introduction}

To compute the low-energy Lagrangian of the string theory compactified on a CY 
 manifold \cite {CHSW}, one needs to know the Special geometry of the corresponding CY moduli space 
\cite{ Distances, S, Rolling, CO}.

More precisely, the effective Lagrangian of the vector multiplets in the superspace contains
 $h^{2,1}$ supermultiplets. Scalars from these multiplets take value in the
 target space $\MM$, which is a moduli space
 of complex structures on a CY manifold and is a special K\"ahler manifold.
 Metric $G_{a\bar{b}}$ and Yukawa couplings $\kappa_{abc}$ 
on this space are given by the following formulae in the special coordinates $z^a$ :
\ba \label{koo}
&G_{a\bar{b}} = \pd_a\pd_{\bar{b}} K, \;\;\;
 e^{-K} = -i \int_X \Omega\wedge\bar{\Omega},\\
&\kappa_{abc} = \int_X \Omega\wedge\pd_a\pd_b\pd_c\Omega = \frac{\pd^3 F}{\pd z^a \pd z^b \pd z^c},
\ea
where
\be
z^a = \int_{A_a} \Omega, \; \frac{\pd F}{\pd z^a} = \int_{B^a} \Omega
\ee
are the period integrals of the holomorphic volume form $\Omega$ on $X$.
Here $A_a$ and $B^a$ form the symplectic basis in $H_3(X,\ZZ)$.

We can rewrite the expression~\eqref{koo} for the K\"ahler potential using the periods as
\be \label{ksympl}
e^{-K} = -i \Pi \Sigma \Pi^{\dagger}, \; \Pi = (\pd F, \; z), 
\ee
where matrix $(\Sigma)^{-1}$ is an intersection matrix of cycles $A_a, \; B^a$ equal to the
symplectic unit.

The computation of periods in the symplectic basis appears to be very non-trivial. It was firstly  performed for the case of the quintic CY manifold in the distinguished paper  \cite{COGP}.\\
Here I present an alternative approach to the computation of  K\"ahler potential for the case
where CY manifold is given by a hypersurface $W(x, \phi) = 0$ in a weighted projective space.
The approach is based on the connection of  CY manifold with a Frobenius ring which arises on the deformations of the singularity defined by the superpotential   $W_0(x)$  \cite{ LVW,  Mart, VW}.
  
Let a CY manifold $X$ be given as a solution of an equation 
\be
W(x, \phi) = W_0(x) + \sum_{s=1}^{h^{2,1}} \phi_s e_s(x)  = 0
\ee
in  some weighted projective space, where $W_0(x)$ is a quasihomogeneous function in $\CC^5$
 of degree $d$ that defines an isolated singularity at $x=0$. 
The monomials $e_s(x) $ also have  degree $d$ and are in a  correspondence to deformations of
 the complex structure of  $X$.

Polynomial $W_0(x)$ defines a Milnor ring $R_0$. Inside  $R_0$  there exists a subring $R^Q_0$  which
 is invariant under the action of the so-called quantum symmetry group $Q$ that
acts on $\CC^5$ diagonally, and preserves $W(x, \phi)$. 
In many cases 
$\dim R_0^Q = \dim H^3(X)$ and the ring itself has a Hodge structure $R_0^Q = (R_0^Q)^0 
\oplus (R_0^Q)^1 \oplus (R_0^Q)^2 \oplus (R_0^Q)^3$ in correspondence with  the elements 
of degrees $0,d,2d,3d$.

Another important group is the subgroup of phase symmetries $G$, which acts diagonally on $\CC^5$, commutes with  the quantum symmetry $Q$  and preserves $W_0(x)$. 
It acts naturally on the invariant ring $R_0^Q$, and this action respects
 the Hodge decomposition of  $R_0^Q$. 
This allows to choose a basis $e_{\mu}(x)$ in each of the Hodge decomposition components
 of $R_0^Q$ to be eigenvectors for the   $G$  group action.

On the  ring  $R_0^Q$ we introduce  the invariant pairing  $\eta$. The pairing turns the ring to a 
Frobenius algebra \cite{Dub}. The pairing   $\eta$  plays an important for  our construction of 
the explicit expression  for the volume of the Calabi-Yau manifold.

Using the invariant ring $R_0^Q$ and differentials $D_{\pm} = \dd \pm \dd W_0\wedge$ we construct
two $Q-$invariant cohomology groups  $H^5_{D_{\pm}}(\CC^5)_{inv}$. 
These groups inherit the Hodge
structure from $R_0^Q$. We can choose in $H^5_{D_{\pm}}(\CC^5)_{inv}$ the  eigenbasises $e_{\mu}(x) \, \dd^5 x$ which are also invariant under the phase symmetry action.

As shown  in~\cite{Candelas}, elements of these cohomology groups are in correspondence with the harmonic
forms of $H^3(X)$. 
This isomorphism allows to define  the antilinear involution  $*$ 
on the invariant cohomology $H^5_{D_{\pm}}(\CC^5)_{inv}$ 
that corresponds to the complex conjugation on the space of the harmonic forms of $H^3(X)$. 

It turns out, that in the basis $e_{\mu}(x)$ it reads
\be \label{conj}
*e_{\mu}(x)\,\dd^5 x =  M_{\mu}^{\nu} e_{\nu} (x) \, \dd^5 x ,
 \, M_{\mu}^{\nu} = \delta_{e_{\mu}\cdot e_{\nu}, e_{\rho}} A^{\mu} 
\ee
where $e_{\rho}(x)$ is the unique  element of degree $3d$ in $R_0^Q$, and
$ \delta_{e_{\mu} \cdot  e_{\nu}, e_{\rho}}$ is 1 if $e_{\mu}\cdot e_{\nu} =  e_{\rho}$
and  $0$ otherwise.

Having $H^5_{D_{\pm}}(\CC^5)_{inv}$ we define the relative invariant homology subgroups
$\mathscr{H}_{5}^{\pm,inv} := H_5(\CC^5, W_0 = L,\;\mathrm{Re}L \to \pm \infty)_{inv}$ inside the relative homology groups
$H_5(\CC^5, W_0 = L,\;\mathrm{Re}L \to \pm \infty)$. 
To do this  we will  use the oscillatory integrals and their pairing 
with elements of $H^5_{D_{\pm}}(\CC^5)_{inv}$  .
 Using  this pairing we define a cycle  $\Gamma^{\pm}_{\mu}$  
in the basis of relative invariant homology to  be  dual to  $e_{\mu}(x) \, \dd^5 x$.

At last we define periods $\sigma^{\pm}_{\mu}(\phi)$ to be oscillatory integrals over the basis 
of cycles $\Gamma^{\pm}_{\mu}$. They are equal to periods of the holomorphic volume form $\Omega$
on $X$ in a special basis of cycles of $H_3(X, \CC)$ with complex coefficients. 

It follows from  the phase symmetry invariance that in the chosen basis of
 cycles $\Gamma^{\pm}_{\mu}$ the formula for K\"ahler potential has the  diagonal form:
\be
e^{-K(\phi)} = \sum_{\mu}  (-1)^{|\nu|}\sigma^+_{\mu}(\phi) A^{\mu} \overline{\sigma^-_{\mu}(\phi)}.
\ee
On the other hand, as
 shown in~\cite{AKBA1}, matrix $A = \mathrm{diag}\{A^{\mu}\}$ is equal to the product of the matrix
 of the invariant pairing $\eta$ 
in the Frobenius algebra $R_0^Q$  and the real structure matrix $M$ such that 
\be
e^{-K(\phi)} = \sum_{\mu,\nu}  \sigma^+_{\mu}(\phi) \eta^{\mu\lambda}
\, M^{\nu}_{\lambda} \overline{\sigma^-_{\mu}(\phi)}.
\ee
 The real structure matrix is nothing but  matrix $M$   from~\eqref{conj}.
Using this we are able to explicitly compute the diagonal matrix elements  $A^{\mu}$ and to obtain the explicit expression  for the whole $e^{-K}$.

  \section{The special geometry on the CY moduli space}\label{sec:special}

It was shown in  in \cite{S, Distances, Rolling, CO}  that the moduli space   $\MM$ of complex (or K\"ahler)  structures of a given CY manifold  is a  special  K\"ahler manifold. \\
Namely on $\MM$  there exist so-called special (projective) coordinates
 $z^1 \cdots z^{n+1}$ and a holomorphic homogeneous function $F(z)$ of degree 2 in $z$, called a 
prepotential, such that  the K\"ahler potential $K(z)$ of the moduli space metric  is
 given by 
\be \label{specmet} e^{-K(z)} = \int_X \Omega\wedge \bar{\Omega} =
  z^a \cdot \frac{\pd \bar{F}}{ \pd \bar{z}^{\bar{a}}} - \bar{z}^{\bar{a}} \cdot \frac{\pd F}{ \pd z^{a}} 
\ee

To obtain this formula, we  choose Poincare dual symplectic basises 
 $ \alpha_a,\beta^b \in H^3(X, \ZZ)$ and $A^a,B_b\in H_3(X, \ZZ)$ and define the periods  as

\be
z^a = \int_{A^a} \Omega, \; F_b = \int_{B_b} \Omega.
\ee
Then using the  Kodaira Lemma
\be
\pd_a \Omega = k_a \Omega + \chi_a,
\ee
we can show that
\be
F_a(z) = \frac{1}{2} \pd_a (F(z)),
\ee
where $F(z)= 1/2 z^b F_b(z)$.\\

Therefore, according to the definition~\eqref{specmet} metric
 $G_{a\bar{b}} = \pd_a\bar{\pd}_{\bar{b}} \, K(z)$ is a special  K\"ahler metric 
with prepotential $F(z)$ and with the special coordinates given by the period vector
\be
\Pi = \begin{pmatrix} F_{\alpha}, z^{b} \end{pmatrix}
\ee
we write the expression for the  K\"ahler potential as 
\be \label{symplkahler}
e^{-K(z)} = \Pi_\mu \Sigma^{\mu \nu} \bar{\Pi}_\nu,
\ee 
where $\Sigma$ is a symplectic unit, which is an inverse intersection matrix for cycles
 $A^a$ and $ B_b$.

Using formula~\eqref{symplkahler}, we can rewrite this expression in a  basis of periods 
defined as integrals over arbitarary basisis of cycles $q_\mu \in H_3(X, \ZZ)$

\be
\omega_{\mu} = \int_{q_\mu} \Omega ~ .
\ee
Such that

\be\label{kahpot2}
e^{-K} = \omega_{\mu} C^{\mu\nu} \bar{\omega}_{\nu},
\ee
where $ C^{\mu\nu}$ is the inverse  marix of the intersection of the cycles  $q_\mu$.

So to find   the K\"ahler potential, we must compute the periods over a basis of cycles
 on CY manifold   and find their intersection matrix.

\section{Hodge structure on the middle cohomology of the quintic}

Now let us specialize to the case where $X$ is a quintic threefold:
\be
X = \{(x_1: \cdots : x_5) \in \PP^4 \; | \; W(x, \phi) = 0 \},
\ee
and
\be \label{W}
W(x, \phi) = W_0(x) + \sum_{t=0}^{100} \phi_t e_t(x), \; W_0(x) = x_1^5+x_2^5+x_3^5+x_4^5+x_5^5  
\ee
and $e_t(x)$ are the degree 5 monomials such that each variable has the power that is a
 non-negative integer  less then four.
Let us denote monomials $e_t(x) = x_1^{t_1} x_2^{t_2}x_3^{t_3} x_4^{t_4}x_5^{t_5}$ by its degree
vector $t = (t_1, \cdots,  t_5)$. Then there are precisely 101 of such monomials, which can be 
divided into $5$ sets in respect to the permutation group $S_5$: 
$(1,1,1,1,1),$ $(2,1,1,1,0),$ $(2,2,1,0,0),$ $(3,1,1,0,0),$ $(3,2,0,0,0)$. 
In these groups there are correspondingly 1, 20, 30, 30, 20 different
monomials.
We denote $e_0(x) := e_{(1,1,1,1,1)}(x) = x_1x_2x_3x_4x_5$ to be the so-called
 fundamental monomial, which will be somewhat distinguished in our picture. \\

For this CY $\dim H_3(X) = 204$ and period integrals have the form
\be
\omega_{\mu}(x) = \int_{q_{\mu}} \frac{x_5 \, \dd x_1\dd x_2\dd x_3}{\pd W(x, \phi)/\pd x_4} = 
\int_{Q_{\mu}} \frac{\dd x_1 \cdots \dd x_5}{W(x, \phi)},
\ee
where $q_{\mu}\in H_3(X,\ZZ)$  and the corresponding cycles 
$Q_{\mu} \in H_5(\CC^5 \backslash (W(x, \phi) = 0), \ZZ)$. 

Cohomology groups of the K\"ahler manifold $X$ possess a Hodge structure $H^3(X) = H^{3,0}(X)\oplus
H^{2,1}(X)
\oplus H^{1,2}(X)\oplus H^{0,3}(X)$.
Period integrals measure variation of the Hodge structure on $H^3(X)$ as the complex structure
 on $X$ varies with $\phi$.

 This Hodge structure variation is in correspondence with a Frobenius ring which we
will now describe.\\

\section{Hodge structure on the invariant Milnor ring}

Now we will consider $W_0(x)$ as an isolated  singularity in $\CC^5$
 and the  associated  with it Milnor ring
\be
R_0 = \frac{\CC[x_1, \cdots, x_5]}{\langle\pd_i W_0 \rangle}.
\ee
We can choose its elements to be unique smallest degree polynomial representatives.
For the quintic threefold $X$ its Milnor ring $R_0$ is generated as a vector space
 by monomials where each variable 
has degree less than four, and $\dim R_0 = 1024$. 

Since the polynomial $W_0(x)$ is homogeneous one of the fifth degree it follows that
$W_0(\alpha  x_1, \ldots, \alpha x_5) = W_0(x_1, \ldots, x_5)$ for $\alpha^5 = 1$. This action preserves $W_0(x)$ and is trivial in the corresponding projective space and on $X$. Such a group with this action is called 
a~\textit{quantum symmetry} $Q$, in our case $Q \simeq \ZZ_5$. 
$Q$ obviously acts on the Milnor ring $R_0$. \\
We define a subring $R_0^Q$  to be a $Q$-invariant part of the Milnor ring
\be
R_0^Q := \{e_{\mu}(x) \in R_0 \; | \; e_{\mu}(\alpha x)  = e_{\mu}(x)\}, \; \alpha^5 = 1.
\ee
$R_0^Q$  is multiplicatively
generated by 101 fifth-degree monomials $e_t(x)$ from~\eqref{W} and 
consists of elements of degree $0,5,10$ and $15$. The dimensions of the corresponding subspaces
are $1,101,101$ and $1$.\\
 This degree filtration defines a Hodge structure on $R^Q_0$. 
Actually,  $R_0^Q$ is isomorphic to $H^3(X)$ and this isomorphism  sends the degree filtration on   $R_0^Q$ to the Hodge filtration on $H^3(X)$ \cite{Candelas}.\\
Let us denote $\chi^i_{\bar{j}} = g^{i\bar{k}} \, \chi_{\bar{k}\bar{j}}$ as an extrinsic curvature tensor  
and $ g_{i\bar{k}}$ is a metric  for the hypersurface $W(x, \phi) = 0$ in $\PP^4$. 
Then the isomorphism above can be written as a map from $R_0^Q$ to closed differential
forms in $H^3(X)$:

\ba \label{chi}
&1 \to \Omega_{ijk} \in H^{3,0}(X), \\
&e_{\mu}(x) \to e_{\mu}(x(y)) \, \chi^l_{\bar{i}} \, \Omega_{ljk} \in H^{2,1}(X) \text{ if } |\mu| = 5, \\
&e_{\mu}(x) \to e_{\mu}(x(y)) \, \chi^l_{\bar{i}} \, \chi^m_{\bar{j}} \, \Omega_{lmk} \in H^{1,2}(X) \text{ if }
 |\mu| = 10, \\
&e_{\rho}(x) = x_1^3x_2^3x_3^3x_4^3x_5^3 \to \chi^l_{\bar{i}} \, \chi^m_{\bar{j}} \chi^p_{\bar{k}} \, \Omega_{lmp} = \kappa \bar{\Omega} \in H^{0,3}(X)
\ea
 The details of this map can be found in~\cite{Candelas}. We also introduce the notation $e_{\mu}(x)$  for elements of the monomial basis of $R^Q_0$,
where  $\mu = (\mu_1, \cdots, \mu_5), \; \mu_i \in \ZZ_+^5 ,\; e_{\mu}(x) = \prod_i x_i^{\mu_i}$ and  the degree of $e_{\mu}(x)$  $\mu = \sum \mu_i$ is  equal to zero module $5$. 
In particular, $\rho = (3,3,3,3,3),$ that is $e_{\rho}(x)$ is the unique degree 15 element of $R_0^Q$.\\
The phase symmetry group  $\ZZ^5_5$  acts diagonally on $\CC^5$:
$\alpha \cdot (x_1, \cdots, x_5) = (\alpha_1 x_1, \cdots, \alpha_5 x_5), \; \alpha_i^5=1$.
 This action preserves
$W_0 = \sum_i x_i^5$. The mentioned above quantum symmetry $Q$ is a diagonal
 subgroup of the phase symmetries. Basis $\{e_{\mu}(x)\}$ consits of the  eigenvectors of the phase symmetry and each $e_{\mu}(x)$ has a unique weight. 
Note that the action of the  phase symmetry preserves  the Hodge decomposition.\\
Another important fact is  that on the invariant ring $R_0^Q$ there exists a natural invariant pairing  turning it into a Frobenius algebra~\cite{Dub}:
\be
\eta_{\mu\nu} = \Res \frac{e_{\mu}(x) \, e_{\nu}(x)}{\prod_i \pd_i W_0(x)}.
\ee
Up to an irrelevant constant for the monomial basis it is $\eta_{\mu\nu} = \delta_{\mu+\nu, \rho}$.
This pairing
plays a crucial role in our construction.\\
Let us  introduce a couple of Saito differentials as in ~\cite{AKBA1} on differential forms on $\CC^5: 
\; D_{\pm} = \dd \pm \dd W_0(x) \wedge$.
They define two cohomology groups $H^*_{D_{\pm}}(\CC^5)$. The cohomologies are only nontrivial
in the top dimension $H^5_{D_{\pm}}(\CC^5) \overset{J}{\simeq} R_0$. The isomorphism $J$ has an explicit description
\be
J(e_{\mu}(x)) = e_{\mu}(x) \, \dd^5 x, \; e_{\mu}(x) \in R_0.
\ee 
We see, that $Q = \ZZ_5$ naturally acts on $H^5_{D_{\pm}}(\CC^5)$ and $J$ sends
 the elements of $Q$-invariant ring $R_0^Q$ to $Q$-invariant subspace $H^5_{D_{\pm}}(\CC^5)_{inv}$. 
Therefore, the latter space obtains the Hodge structure  as well.
Actually, this Hodge structure  naturally corresponds  to the Hodge structure on $H^3(X)$.

 The complex conjugation  acts 
on $H^3(X)$ so that $\overline{H^{p,q}(X)} = H^{q,p}(X)$, in particular
 $\overline{H^{2,1}(X)} = H^{1,2}(X)$. Through the isomorphism  between
$R_0^Q$ and $H^3(X)$ the complex conjugation acts also on the elements of the
ring $R_0^Q\; $ as  $* e_{\mu}(x) = p_{\mu} e_{\rho - \mu}(x),$ where  $p_{\mu} p_{\rho - \mu}=1$ and
$p_{\mu}$ is a constant to be determined. In particular, differential form built from the linear combinations  $e_{\mu}(x) + p_{\mu} e_{\rho - \mu}(x)
\in H^3(X, \RR)$ is real.

\section{Oscillatory representation and computation periods  $\sigma_{\mu}(\phi)$}

Relative homology groups $H_5(\CC^5, W_0 = L,\;\mathrm{Re}L \to \pm \infty)$ have a natural 
pairing with $Q$-invariant cohomology groups $H^5_{D_{\pm}}(\CC^5)_{inv}$ defined as
\be
\langle e_{\mu}(x)\dd^5 x, \Gamma^{\pm} \rangle = \int_{\Gamma^{\pm}}
 e_{\mu}(x) e^{\mp W_0(x)} \dd^5 x, \;
  H_5(\CC^5, W_0 = L,\;\mathrm{Re}L \to \pm \infty).
\ee
Using this we introduce two  $Q$-invariant homology 
groups\footnote{We are grateful to V. Vasiliev for explaining to  us the details 
about these homology groups and their connection with the middle homology of $X$.} 
$\mathscr{H}_5^{\pm,inv}$ as
 quotient of $H_5(\CC^5, W_0 = L,\;\mathrm{Re}L \to \pm \infty)$ with respect to the subgroups orthogonal
to $H^5_{D_{\pm}}(\CC^5)_{inv}$. 
Now we introduce  basises $\Gamma^{\pm}_{\mu}$ in the homology groups  
$\mathscr{H}_5^{\pm,inv}$ using the duality with the basises 
in  $H^5_{D_{\pm}}(\CC^5)_{inv}$:
\be
\int_{\Gamma^{\pm}_{\mu}}
 e_{\nu}(x) e^{\mp W_0(x)} \dd^5 x = \delta_{\mu\nu}
\ee
and the corresponding periods
\ba \label{sig1}
&\sigma_{\alpha\mu}^{\pm}(\phi) := \int_{\Gamma^{\pm}_{\mu}}
 e_{\alpha}(x) e^{\mp W(x, \phi)} \dd^5 x, \\
 &\sigma_{\mu}^{\pm}(\phi) := \sigma_{0\mu}^{\pm}(\phi)
\ea
which are understood as series expansions in $\phi$ around zero.\\
 The periods 
$\sigma_{\mu}^{\pm}(\phi)$ satisfy the same differential equation as periods $\omega_{\mu}(\phi)$
of the holomorphic volume form on $X$. Moreover, these sets of periods span same subspaces as
functions of $\phi$. Therefore we can define cycles $Q^{\pm}_{\mu} \in \mathscr{H}_5^{\pm,inv}$ such that
\be \label{qcyc}
\int_{Q^{\pm}_{\mu}}
  e^{\mp W(x, \phi)} \dd^5 x =  \int_{q_{\mu}}\Omega = \int_{Q_{\mu}}\frac{\dd^5 x}{W(x, \phi)}.
\ee
So the periods $\omega_{\alpha\mu}^{\pm}(\phi)$ are given by the integrals over these cycles analogous to~\eqref{sig1}.

With these notations the idea of computation of periods~\cite{BB}
\be \label{sigma1}
\sigma^{\pm}_{\mu}(\phi) = \int_{\Gamma^{\pm}_{\mu}} 
 e^{\mp W(x, \phi)} \, \dd^5 x
\ee
can be stated as follows. \\
 To explicitly compute $\sigma^{\pm}_{\mu}(\phi)$, first  we expand the exponent in the
 integral~\eqref{sigma1} in $\phi$ representing $W(x,\phi) = W_0(x) + \sum_s \phi_s e_s(x)$
\be \label{sigma2}
\sigma^{\pm}_{\mu}(\phi) = \sum_m  \left(\prod_s\frac{(\pm\phi_s)^{m_s}}{ m_s!}\right) \int_{\Gamma^{\pm}_{\mu}} 
 \prod_s e_{s}(x)^{m_s} \, e^{\mp W_0(x)} \, \dd^5 x.
\ee
We note, that $\sigma^-_{\mu}(\phi) = (-1)^{|\mu|}\sigma^+_{\mu}(\phi),$
so we focus on $\sigma_{\mu}(\phi) := \sigma^{+}_{\mu}(\phi).$\\
For each of the summands in~\eqref{sigma2} the form
 $\prod_s e_{s}(x)^{m_s} \, \dd^5 x$ belongs to $H^5_{D_{\pm}}(\CC^5)_{inv},$
because it is $Q-$invariant.
Therefore, we can expand it in the
 basis $e_{\mu}(x) \, \dd^5 x \in
H^5_{D_{\pm}}(\CC^5)_{inv}.$ Namely we  can find such a polynomial $4-$form $U,$ that
\be \label{sigma3}
 \prod_s e_{s}(x)^{m_s} \, \dd^5 x =
 \sum_{\nu} C_{\nu}(m) \, e_{\nu}(x) \, \dd^5 x + D_{+} U  .
\ee
In result  we obtain for the integral in~\eqref{sigma2} 
\be \label{sigma4}
\int_{\Gamma^{\pm}_{\mu}} 
 \prod_s e_{s}(x)^{m_s} \, e^{\mp W_0(x)} \, \dd^5 x = C_{\mu}(m).
\ee

So from~\eqref{sigma2} we have
\be \label{qsigma1}
\sigma_{\mu}(\phi) = \sum_m \left(\prod_s\frac{\phi_s^{m_s}}{ m_s!}\right) \int_{\Gamma^{+}_{\mu}} 
 \prod_{s,i}  x_i^{ \sum_{s} m_s s_i } \, e^{- W_0(x)} \, \dd^5 x.
\ee
We can  rewrite the sum in the exponent of $x_i$ as $  \sum_{s}m_s s_i = 5 n_i + \nu_i, \; \nu_i < 5$.\\
Therefore we need to compute the coefficients $c^m_{\nu}$ in the equations 
\be
\prod x_i^{5 n_i + \nu_i} \, \dd^5 x =
 \sum_{\nu} c^m_{\nu} \, e_{\nu}(x) \, \dd^5 x + D_{+} U.
\ee
Note that 
\begin{multline} \label{rec}
D_+ \left(\frac{1}{5}x_1^{5n + k-4} \, f(x_2, \cdots, x_5) \, \dd x_2 \wedge \cdots \wedge \dd x_5
\right) = \\ =
\left[x_1^{5n+k} + \left(n+\frac{k-4}{5}\right) x_1^{5(n-1) + k} \right]  \, f(x_2, \cdots, 
x_5) \, \dd^5 x
\end{multline}
Therefore in $D_+$ cohomology we have
\be \label{rec1}
\prod_i x_i^{5 n_i + \nu_i} \, \dd^5 x = 
-\left(n_1+\frac{\nu_1-4}{5}\right) x_1^{5(n_1-1) + \nu_1}
\prod_{i=2}^5 x_i^{5n_i+\nu_i} \, \dd^5 x, \; \nu_i < 5.
\ee
By induction we obtain
\be \label{rec1}
\prod_i x_i^{5 n_i + \nu_i} \, \dd^5 x = 
(-1)^{\sum_i n_i}\prod_i \left(\frac{\nu_i+1}{5}\right)_{n_i}
\prod_i x_i^{\nu_i} \, \dd^5 x, \; \nu_i < 5.
\ee
where $(a)_{n}=\Gamma(a+n)/\Gamma(a)$.\\
Using~\eqref{rec} once again, we see that if any $\nu_i = 4$ then the differential form is trivial and
 the integral is zero. Hence, rhs of~\eqref{rec1} is proportional to $e_{\nu}(x)$ and
 gives the desired 
expression. Plugging~\eqref{rec1} into~\eqref{qsigma1} and integrating over $\Gamma^+_{\mu}$ we obtain the answer
\be
\sigma_{\mu}(\phi)=\sigma_{\mu}^+(\phi) =
 \sum_{n_i\ge 0} \prod_i \left(\frac{\mu_i+1}{5}\right)_{n_i} 
\sum_{m \in \Sigma_n } \prod_s\frac{\phi_s^{m_s}}{ m_s!} ,
\ee
where 
\be
\Sigma_n = \{m\;|\;\sum_s m_s s_i = 5n_i+\mu_i\}
\ee
Further we will also use the periods with slightly different normalization, which turn
 out to be convenient
\be \label{sigmaquint}
\hat{\sigma}_{\mu}(\phi) = \prod_i \Gamma\left(\frac{\mu_i+1}{5}\right)\sigma_{\mu}(\phi)=
 \sum_{n_i\ge 0} \prod_i \Gamma\left(n_i+\frac{\mu_i+1}{5}\right) 
\sum_{m \in \Sigma_n } \prod_s\frac{\phi_s^{m_s}}{ m_s!}. 
\ee

\section{Computation of the K\"ahler potential}
 Pick any basis $Q^{\pm}_{\mu}$ of cycles with integer or real coefficients as in~\eqref{qcyc}.
Then for the K\"ahler potential we have the formula
\be \label{eqK1}
e^{-K} = \omega^+_{\mu}(\phi) C^{\mu\nu} \overline{\omega^-_{\nu}(\phi)}
\ee
in which the matrix $C^{\mu\nu}$ is related with the  Frobenius pairing $\eta$ as
\be \label{eqeta}
\eta_{\alpha\beta} = \omega^+_{\alpha\mu}(0)C^{\mu\nu} \omega^-_{\beta\nu}(0).
\ee
The derivation of the last relation is given in~\cite{CV, Chiodo}.

Let also $T^{\pm}$ be the  matrix that connects   the cycles
 $Q^{\pm}_{\mu}$ and $ \Gamma^{\pm}_{\nu}$.\\
 That is   
$$Q^{\pm}_{\mu} = (T^{\pm})^{\nu}_{\mu} \Gamma^{\pm}_{\nu}$$. 
Then  $M = (T^-)^{-1}\overline{T^-}$ is a real structure matrix, that is $M\bar{M} = 1$ and by construction $M$ doesn't depend on the choice
of basis $Q^{\pm}_{\mu}.$ $M$ is only defined by our choice of $\Gamma^{\pm}_{\mu}$.

In~\cite{AKBA1} we deduced from~\eqref{eqK1} and~\eqref{eqeta}  the formula
\be \label{eqK}
e^{-K(\phi)} = \sigma^+_{\mu}(\phi) \eta^{\mu\lambda} M^{\nu}_{\lambda} \overline{\sigma^-_{\nu}(\phi)} = \sigma_{\mu} A^{\mu\nu}
\overline{\sigma_{\nu}},
\ee
where $\eta^{\mu\nu} = \eta_{\mu\nu} = \delta_{\mu, \rho-\nu}$.

Now we show that the matrix $A^{\mu\nu}$ in~\eqref{eqK} is diagonal. 
 To do this we extend the action of the phase symmetry group to the action $\aA$ on the
 parameter space $\{ \phi_s \}$ such
that $W = W_0 + \sum_s \phi_s e_s(x)$ is invariant under this new action. It easy to see that each $e_s(x)$ has an unique weight under this group action. 
Action $\aA$ can be compensated using the coordinate tranformation and therefore
 is trivial on the moduli space of the quintic (implying that point $W=W_0$ is an
 orbifold point of the moduli space).

In particular, $e^{-K} = \int_X \Omega\wedge\bar{\Omega}$ is $\aA$  invariant. Consider
\be
e^{-K} = \sigma_{\mu} A^{\mu\nu} \overline{\sigma_{\nu}}
\ee
as a series in $\phi_s$ and $ \; \overline{\phi_t}$. 
Each monomial has a certain weight under $\aA$ . 
For the series to be invariant, each monomial must have weight 0. 
But weight of $\sigma_{\mu} \overline{\sigma_{\nu}}$ equals
to $\mu - \nu$ and due to non-degeneracy of weights of $\sigma_{\mu}$
 only the ones with $\mu = \nu$ have weight zero.
Thus,~\eqref{eqK} becomes
\be
e^{-K} = \sum_{\mu} A^{\mu} |\sigma_{\mu}(\phi)|^2.
\ee

Moreover, the matrix $A$ should be real and, because $A = \eta \cdot M, \; M\bar{M}=1$ and $\eta_{\mu\nu} = \delta_{\mu+\nu,\rho},$ we have
 \be \label{mon}
A^{\mu} \, A^{\rho - \mu} = 1.
\ee

\paragraph{Monodromy considerations}
To fix finally the real numbers $A^{\mu}$ we use monodromy invariance of $e^{-K}$ around 
$\phi_0=\infty.$  Pick some $t = (t_1,t_2,t_3,t_4,t_5)$ with $ \; |t|=5$ and let $\phi_s|_{s\ne t,0} = 0,$ . We will  consider only the first order in $\phi_t$.\\
 Then the condition that period $\sigma_{\mu}(\phi)$ contains only non-zero summands
 of the form $\phi_0^{m_0} \, \phi_t$ implies that $\mu = t + const\cdot (1,1,1,1,1)$ mod 5.
 For each $t$
from the table below the only such possibilities are $\mu = t$ and $\mu = \rho - t' = (3,3,3,3,3) - t',$
 where $t'$ denotes a vector obtained from $t$ by permutation 
(written explicitly in the table below) of its coordinates.

Therefore, in this setting~\eqref{eqK} becomes
\be
e^{-K} = \sum_{k=0}^3 a_k |\hat{\sigma}_{(k,k,k,k,k)}|^2 + a_t|\hat{\sigma}_t|^2 + a_{\rho-t'} |\hat{\sigma}_{\rho-t'}|^2 + O(\phi_t^2),
\ee
here we use periods $\hat{\sigma}$ from~\eqref{sigmaquint} and denote  $a_t = A^t /\prod_i \Gamma((t_i+1)/5)^2$. And the coefficients   $a_k, \; k=0,1,2,3$ are already known from \cite {COGP}.
 This expression has to be monodromy invariant under
the transport of $\phi_0$ around $\infty$. From the formula~\eqref{sigmaquint} we have
\ba
&F_1 = \hat{\sigma}_{k}(\phi_t, \phi_0) = g_t \phi_k \, F(a,b;a+b \, |\, (\phi_0/5)^5) + O(\phi_t^6), \\
&F_2 = \hat{\sigma}_{\rho-t'}(\phi_t, \phi_0) = g_{\rho-t'} \phi_t \, \phi_0^{1-a-b} \, F(1-a,1-b;2-a-b\,|\, (\phi_0/5)^5) + O(\phi_t^6), 
\ea
where $g_t, \; g_{\rho-t'}$ are some constants. Explicitly for all different labels $t$\\
\begin{center}
\begin{tabular}{| l | l | l |} \label{tab1}
  $t$ & $\rho-t'$ & (a, b)  \\ \hline
  (2,1,1,1,0) & (3,2,2,2,1) & (2/5,2/5)  \\
  (2,2,1,0,0) & (3,3,2,1,1) & (1/5,3/5)  \\
  (3,1,1,0,0) & (0,3,3,2,2) & (1/5,2/5)  \\
  (3,2,0,0,0) & (1,0,3,3,3) & (1/5,1/5)  \\
\end{tabular}
\end{center}
When $\phi_0$ goes around infinity
\be
\begin{pmatrix}
F_1 \\
F_2
\end{pmatrix} = 
B \cdot\begin{pmatrix}
F_1 \\
F_2
\end{pmatrix},
\ee
where 
\be
B = \frac{1}{i s(a+b) }\begin{pmatrix}
c(a-b) - e^{i\pi (a+b)} & 2 s(a)s(b) \\
2 e^{2\pi i (a+b)} s(a)s(b) & e^{\pi i(a+b)} [e^{2\pi i a} + e^{2\pi i b}-2]/2
\end{pmatrix}.
\ee
Here $c(x) = \cos(\pi x), \; s(x) = \sin(\pi x)$. It is straightforward to show the following 
\begin{prop}
\be
a_t |\hat{\sigma}_t|^2 + a_{\rho-t'} |\hat{\sigma}_{\rho-t'}|^2 = 
a_t \prod_i \Gamma\left(\frac{t_{i}+1}{5}\right)^2 |\sigma_t|^2 +
 a_{\rho-t'}\prod_i \Gamma\left(\frac{4-t_{i}}{5}\right)^2 |\sigma_{\rho-t'}|^2 
\ee
is $B$-invariant iff $a_t = -a_{\rho-t'}$. 
\end{prop}
Due to symmetry we have $a_{\rho-t'} = a_{\rho-t}$ in each case. From~\eqref{mon} 
it follows that the product of the coefficients at $|\sigma_{\mu}|^2$ and
 $|\sigma_{\rho-\mu}|^2$ in the expression for $e^{-K}$ should be 1:
\be
A^{\rho-t'} \cdot A^t=a_{\rho-t'} \cdot a_t \prod_i \Gamma\left(\frac{t_{i}+1}{5}\right)^2 
\Gamma\left(\frac{4-t_{i}}{5}\right)^2= 1.
\ee
Due to reflection formula $a_t = \pm\prod_i\sin(\pi (t_i+1)/5)$ up to a common factor of $\pi$.
The sign turns out to be minus for K\"ahler metric to be positive definite in the origin.
Therefore 
\be
A^{\mu} = (-1)^{\deg (\mu)/5} \prod{\gamma\left(\frac{\mu_i+1}{5}\right)}.
\ee

 Finally the K\"ahler potential becomes 
\be \label{eqKq}
e^{-K(\phi)} = \sum_{\mu=0}^{203}(-1)^{\deg (\mu)/5} \prod{\gamma\left(\frac{\mu_i+1}{5}\right)}
 |\sigma_{\mu}(\phi)|^2, 
\ee
where $ \gamma(x) = \frac{\Gamma(x)}{\Gamma(1-x)}$.

\section{Real structure on the cycles $\Gamma^{\pm}_{\mu}$}

Let cycles $\gamma_{\mu} \in H_3(X)$ be the images of cycles $\Gamma^+_{\mu}$ under the isomorphism
$\mathscr{H}_5^{+,inv} \simeq H_3(X)$. \\
Complex conjugation sends $(2,1)$-forms to $(1,2)$-forms. Similarly it
extends to a mapping on the dual homology cycles $\gamma_{\mu}$.

\begin{lemma} Conjugation of homology classes has the
following form: $ *\gamma_{\mu} = p_{\mu}\gamma_{\rho - \mu},$
where  $\rho = (3,3,3,3,3)$ is a unique maximal degree element in the Milnor ring.\\
\end{lemma}

\begin{proof} We perform a proof for the cohomology classes represented by differential forms.
For one-dimensional $H^{3,0}(X)$ and $H^{0,3}(X)$ it is obvious. Let 
\be
\Omega_{2,1} := e_{t}(x)  \, \chi^l_{\bar{i}} \, \Omega_{ljk} \in H^{2,1}(X).
\ee
Any element from $H^{1,2}(X)$ is representable by a degree 10 polynomial  $P(x)$ as follows from~\eqref{chi}  as
\be
\overline{\Omega_{2,1}} = \Omega_{1,2}:=P(x)  \, \chi^l_{\bar{i}} \ 
 \chi^m_{\bar{j}} \, \Omega_{lmk} \in H^{1,2}(X).
\ee
The group of phase symmetries modulo common factor acts by isomorphisms on $X$. Therefore,
it also acts on the differential forms. Lhs and rhs of the previous equation should have the same
weigth under this action,
and weight of the lhs is equal $-t$ modulo $(1,1,1,1,1)$.
It follows that $P(x) = p_{t} \, e_{\rho - t}(x)$ with some constant $p_t$.
\end{proof}

Using this lemma and applying the complex conjugation of cycles to the formula~\eqref{eqK}
to obtain
\be
e^{-K} = \sum_{\mu} A^{\mu} |\sigma_{\mu}|^2  = 
\sum_{\mu} p_{\mu}^2 A^{\mu} \;|\sigma_{\rho - \mu}|^2,
\ee
it follows that $A^{\mu} = \pm 1/p_{\mu}.$  Now formula~\eqref{eqKq} implies
\be\label{coeff}
p_{\mu} = \prod{\gamma\left(\frac{4-\mu_i}{5}\right)}.
\ee

\section{Conclusion}

I  am  grateful to K.~ Aleshkin for the interesting collaboration; the talk is based on the  joint work  with K.~ Aleshkin. Also I am thankful  to M. Bershtein, V. Belavin, S.Galkin,  D. Gepner, A. Givental, M. Kontsevich,  A. Okounkov, A. Rosly, V. Vasiliev for the  useful discussions. 
The work  has been performed for FASO budget project No. 0033-2018-0006.

\end{document}